\documentclass[letterpaper,11pt]{article}

\usepackage[scale=0.8]{geometry}
\usepackage{amsmath}
\usepackage{amsthm}
\usepackage{amssymb}
\usepackage{mathtools}
\usepackage{times}
\usepackage{url}
\usepackage{enumerate}

\newcommand{\NC}{\mathrm{NC}}
\newcommand{\BQP}{\mathrm{BQP}}
\newcommand{\PP}{\mathrm{PP}}
\newcommand{\QMA}{\mathrm{QMA}}
\newcommand{\PSPACE}{\mathrm{PSPACE}}
\newcommand{\EXP}{\mathrm{EXP}}
\newcommand{\EXPSPACE}{\mathrm{EXPSPACE}}
\newcommand{\IP}{\mathrm{IP}}

\newcommand{\xorMIP}{{\oplus}\mathrm{MIP}}
\newcommand{\xorMIPstar}{{\oplus}\mathrm{MIP}^*}
\newcommand{\MIPns}{\mathrm{MIP}^{\mathrm{ns}}}
\newcommand{\QIP}{\mathrm{QIP}}

\newcommand{\QQ}{\mathbb{Q}}
\newcommand{\RR}{\mathbb{R}}
\newcommand{\CC}{\mathbb{C}}
\newcommand{\NN}{\mathbb{N}}
\newcommand{\sfS}{\mathsf{S}}
\newcommand{\sfT}{\mathsf{T}}
\newcommand{\sfY}{\mathsf{Y}}
\newcommand{\sfZ}{\mathsf{Z}}
\newcommand{\sfP}{\mathsf{P}}
\newcommand{\rmA}{\mathrm{A}}
\newcommand{\rmB}{\mathrm{B}}
\newcommand{\calH}{\mathcal{H}}
\newcommand{\calX}{\mathcal{X}}
\newcommand{\calY}{\mathcal{Y}}

\newcommand{\poly}{\mathit{poly}}
\newcommand{\trans}{\mathrm{T}}
\newcommand{\problemname}{\textsc}
\newcommand{\yes}{\mathrm{yes}}
\newcommand{\no}{\mathrm{no}}

\DeclareMathSymbol{\mymid}{\mathop}{symbols}{"6A}

\DeclareMathOperator{\Tr}{Tr}
\DeclarePairedDelimiter{\abs}{\lvert}{\rvert}
\DeclarePairedDelimiter{\norm}{\lVert}{\rVert}
\DeclarePairedDelimiter{\floor}{\lfloor}{\rfloor}
\DeclarePairedDelimiter{\bra}{\langle}{\rvert}
\DeclarePairedDelimiter{\ket}{\lvert}{\rangle}

\theoremstyle{plain}
\newtheorem{theorem}{Theorem}
\newtheorem{lemma}[theorem]{Lemma}
\newtheorem{corollary}[theorem]{Corollary}
\theoremstyle{remark}
\newtheorem*{remark}{Remark}

\title{Quantum interactive proofs with weak error bounds}
\author{
  Tsuyoshi Ito\footnotemark[1] \and
  Hirotada Kobayashi\footnotemark[2] \and
  John Watrous\footnotemark[1]}
\date{}

\begin{document}
\maketitle
\renewcommand*{\thefootnote}{\fnsymbol{footnote}}
\footnotetext[1]{Institute for Quantum Computing and School of
  Computer Science, University of Waterloo, Waterloo, Ontario, Canada.}
\footnotetext[2]{Principles of Informatics Research Division, National
  Institute of Informatics, Tokyo, Japan.}
\renewcommand*{\thefootnote}{\arabic{footnote}}

\begin{abstract}
  This paper proves that the computational power of quantum
  interactive proof systems, with a double-exponentially small gap
  in acceptance probability between the completeness and soundness
  cases, is precisely characterized by~$\EXP$, the class of problems
  solvable in exponential time by deterministic Turing machines.
  This fact, and our proof of it, has implications concerning
  quantum and classical interactive proof systems in the setting of
  unbounded error that include the following:
  \begin{itemize}
  \item
    Quantum interactive proof systems are strictly more powerful
    than their classical counterparts in the unbounded-error setting
    unless $\PSPACE=\EXP$, as even unbounded error classical interactive
    proof systems can be simulated in~$\PSPACE$.
  \item
    The recent proof of Jain, Ji, Upadhyay, and Watrous (STOC 2010)
    establishing $\QIP=\PSPACE$ relies heavily on the
    fact that the quantum interactive proof systems defining the class
    $\QIP$ have bounded error.
    Our result implies that some nontrivial assumption on the error
    bounds for quantum interactive proofs is unavoidable to establish
    this result (unless~$\PSPACE=\EXP$).
  \item
    To prove our result, we give a quantum interactive proof system
    for~$\EXP$ with perfect completeness and soundness
    error~$1-2^{-2^\poly}$, for which the soundness error bound is
    provably tight.
    This establishes another respect in which quantum and classical
    interactive proof systems differ, because such a bound cannot hold
    for any classical interactive proof system:
    distinct acceptance probabilities for classical interactive proof
    systems must be separated by a gap that is at least
    (single-)exponentially small.
  \end{itemize}
  We also study the computational power of a few other related
  unbounded-error complexity classes.
\end{abstract}

\section{Introduction}

Interactive proof systems~\cite{Babai85STOC,GolMicRac89SICOMP} are a
central notion in complexity theory.
It is well-known that $\IP$, the class of problems having
single-prover classical interactive proof systems with
polynomially-bounded verifiers, coincides with $\PSPACE$
\cite{Feldman86,LunForKarNis92JACM,Shamir92JACM}, and it was recently
proved that the same characterization holds when the prover and
verifier have quantum computers \cite{JaiJiUpaWat10STOC}.
More succinctly, it holds that
\begin{equation} \label{eq:IP=PSPACE=QIP}
\IP = \PSPACE = \QIP.
\end{equation}
The two equalities in \eqref{eq:IP=PSPACE=QIP} are, in some sense,
intertwined: it is only through the trivial relationship
$\IP\subseteq\QIP$, together with the landmark result
$\PSPACE\subseteq\IP$, that we know $\PSPACE\subseteq\QIP$.
While there exist classical refinements
\cite{Shen92JACM,Meir10-TR10-137} of the original method of Lund,
Fortnow, Karloff, and Nisan \cite{LunForKarNis92JACM} and
Shamir~\cite{Shamir92JACM} used to prove $\PSPACE\subseteq\IP$, there
is no ``short-cut'' known that proves $\PSPACE\subseteq\QIP$ through
the use of quantum computation.

The opposite containments required to prove the two equalities in
the above equation \eqref{eq:IP=PSPACE=QIP} are $\IP\subseteq\PSPACE$
and $\QIP\subseteq\PSPACE$, respectively.
The first containment is usually attributed to Feldman
\cite{Feldman86}, and can fairly be described as being straightforward
to prove.
The standard proof, in fact, gives a polynomial-space algorithm that
computes the optimal acceptance probability for a prover in a
classical interactive proof system \emph{exactly}, with this optimal
probability expressible as some integer divided by $2^k$, where $k$ is
the maximum number of coin-flips used by the verifier.
The proof of the containment $\QIP\subseteq\PSPACE$ given in
\cite{JaiJiUpaWat10STOC}, on the other hand, is more complicated:
it uses known properties of $\QIP$ \cite{KitWat00STOC,MarWat05CC} to
derive a semidefinite programming formulation of it, which is then
approximated in $\PSPACE$ through the use of an algorithm based on the
\emph{matrix multiplicative weights update} method
\cite{AroKal07STOC,WarmuthK06}.
Unlike the standard proof of $\IP\subseteq\PSPACE$, this proof depends
crucially on the bounded-error property of the quantum interactive proof
systems that define $\QIP$.

There must, of course, be alternate ways to prove
$\QIP\subseteq\PSPACE$, and we note that Wu~\cite{Wu10-1004.0264} and
Gutoski and Wu~\cite{GutWu10-1011.2787} have made advances in both
simplifying and extending the proof method of
\cite{JaiJiUpaWat10STOC}.
The main question that motivates the work we present in this paper is
whether the assumption of bounded-error is \emph{required} to prove
$\QIP\subseteq\PSPACE$, or could be bypassed.
Our results demonstrate that indeed \emph{some} assumption on the gap
between completeness and soundness probabilities must be in place to
prove $\QIP\subseteq\PSPACE$ unless $\PSPACE = \EXP$.

To explain our results in greater detail it will be helpful to
introduce the following notation.
Given any choice of functions $m:\NN\rightarrow\NN$ and
$a,b:\NN\rightarrow[0,1]$, where we take $\NN = \{0,1,2,\ldots\}$,
we write $\QIP(m,a,b)$ to denote the class of promise problems%
\footnote{%
  We formulate decision problems as
  \emph{promise problems}~\cite{EveSelYac84IC} because using promise
  problems is more natural than restricting our attention to languages
  in the presence of error bounds.}
$A = (A_{\yes},A_{\no})$ having a quantum interactive proof system%
\footnote{%
  The definitions of quantum computational models based on quantum
  circuits, including quantum interactive proof systems, is
  particularly sensitive to the choice of a gate set in the unbounded
  error setting.
  For our main result we take the standard Toffoli, Hadamard,
  $\pi/2$-phase-shift gate set, but relax this choice for a couple
  of our secondary results.}
with $m(\abs{x})$ messages,
completeness probability at least $a(\abs{x})$ and soundness
error at most $b(\abs{x})$ on all input strings
$x\in A_{\yes}\cup A_{\no}$.
When sets of functions are taken in place of $m$, $a$, or $b$, it is to
be understood that a union is implied.
For example,
\[
  \QIP(\poly,1,1-2^{-\poly}) =
  \bigcup_{m,p\in\poly}\QIP(m,1,1 - 2^{-p}),
\]
where $\poly$ denotes the set of all functions of the
form $p:\NN\rightarrow\NN$ for which there exists a polynomial-time
deterministic Turing machine that outputs $1^{p(n)}$ on input $1^n$
for all $n\in\NN$.
We will also frequently refer to functions of the form
$f\colon\NN\rightarrow [0,1]$ that are polynomial-time computable,
and by this it is meant that a polynomial-time deterministic Turing
machine exists that, on input~$1^n$, outputs a rational number~$f(n)$
in the range~$[0,1]$, represented by a ratio of integers expressed in
binary notation.
Our main result may now be stated more precisely as follows.
\begin{theorem} \label{theorem:qipu}
  It holds that
  \[
  \bigcup_a \QIP(\poly,a,a-2^{-2^{\poly}})
  = \QIP(3,1,1-2^{-2^{\poly}})
  = \EXP,
  \]
  where the union is taken over all polynomial-time computable
  functions~$a\colon\NN\to(0,1]$.
\end{theorem}

The only new relation in the statement of Theorem~\ref{theorem:qipu} is
\begin{equation}
  \EXP \subseteq \QIP(\poly,1,1-2^{-2^{\poly}});
  \label{eq:qipu}
\end{equation}
we have expressed the
theorem in the above form only for the sake of clarity.
In particular, the containment
\[
  \QIP(\poly,1,1-2^{-2^{\poly}})
  \subseteq \QIP(3,1,1-2^{-2^{\poly}})
\]
follows from the fact that
\[
  \QIP(m,1,1 - \varepsilon) \subseteq
  \QIP(3,1,1-\varepsilon/(m-1)^2)
\]
for all $m\in\poly$ and any function
$\varepsilon:\NN\rightarrow[0,1]$, as was proved in
\cite{KemKobMatVid09CC} (or an earlier result of \cite{KitWat00STOC}
with a slightly weaker parameter).
The containment
\[
  \QIP(3,1,1-2^{-2^{\poly}}) \subseteq
  \bigcup_a \QIP(\poly,a,a-2^{-2^{\poly}})
\]
is trivial.
The containment
\[
  \bigcup_a \QIP(\poly,a,a-2^{-2^{\poly}}) \subseteq \EXP
\]
follows from the results of Gutoski and Watrous~\cite{GutWat07STOC},
as a semidefinite program representing the optimal acceptance probability
of a given quantum interactive proof system%
\footnote{The results of Gutoski and Watrous~\cite{GutWat07STOC}
  establish an $\EXP$ upper bound even for interactive proof systems
  with two competing quantum provers, and only mild assumptions on the
  gate set are needed to obtain this containment.
  Namely, the containment holds if the gate set consists of finitely
  many gates and the Choi-Jamio\l kowski representation of each gate
  is a matrix made of rational complex numbers.}
can be solved to an
exponential number of bits of accuracy using an exponential-time
algorithm~\cite{Khachiyan79Doklady,GroLovSch88,NesNem94}.

The new containment~(\ref{eq:qipu}), which represents the main
contribution of this paper, is proved in two steps.
The first step constructs a classical two-prover one-round interactive
proof system with one-sided error double-exponentially close to 1 for
the $\EXP$-complete \problemname{Succinct Circuit Value} problem.
It will be proved that when an instance whose answer is ``no'' is
given to this proof system, provers cannot make the verifier accept
with probability more than double-exponentially close to 1 even if
they are allowed to use a \emph{no-signaling strategy}, i.e., a
strategy that cannot be used for communication between them.
The second step converts this classical two-prover one-round
interactive proof system to a single-prover quantum interactive proof
system without ruining its soundness properties.

Theorem~\ref{theorem:qipu} and its proof have the following three
consequences.
\begin{itemize}
\item
  Unbounded-error quantum interactive proof systems are strictly more
  powerful than their classical counterparts unless $\PSPACE=\EXP$, as
  unbounded-error classical interactive proof systems recognize
  exactly $\PSPACE$.
\item
  The dependence on the error bound in the proof
  in~\cite{JaiJiUpaWat10STOC} is not an artifact of the proof
  techniques, but is a necessity unless~$\PSPACE=\EXP$.
  To be more precise, even though a double-exponential gap is sufficient
  to obtain the $\EXP$ upper bound
  by applying a polynomial-time algorithm for semidefinite programming,
  Theorem~\ref{theorem:qipu} implies that a double-exponential gap
  is not sufficient for the $\PSPACE$ upper bound unless~$\PSPACE=\EXP$.
\item
  Our proof of Theorem~\ref{theorem:qipu} shows that a quantum
  interactive proof system can have a complete\-ness-sound\-ness gap
  smaller than singly exponential, which cannot happen in classical
  interactive proof systems.
  In our quantum interactive proof system for~$\EXP$, the gap is
  double-exponentially small, and this is tight in the sense
  that a dishonest prover can make the verifier accept with
  probability double-exponentially close to 1.
\end{itemize}
We do not know if the double-exponentially small gap in
Theorem~\ref{theorem:qipu} can be improved to one that is
single-exponentially small
by constructing a different proof system.

The two parts of the proof of Theorem~\ref{theorem:qipu} mentioned
above are contained in Sections~\ref{sec:no-signaling-EXP} and
\ref{sec:qipu}.
Some additional results concerning unbounded-error quantum interactive
proof systems are discussed in Section~\ref{sec:additional-results}.

\section{\boldmath A no-signaling proof system for $\EXP$ with a weak
  error bound}
\label{sec:no-signaling-EXP}

As discussed in the previous section, our proof of
the containment~(\ref{eq:qipu}) has two parts.
This section discusses the first part, in which we present a classical
two-prover one-round interactive proof system for an $\EXP$-complete
problem.
The proof system will have perfect completeness and a soundness
error double-exponentially close to 1, even when the provers are
permitted to employ an arbitrary \emph{no-signaling strategy}.
No-signaling strategies, which are defined below, have been considered
previously in \cite{Holenstein09TOC} and \cite{ItoKobMat09CCC}, for instance.

\subsection{Definition of no-signaling proof systems}

In a \emph{(classical) two-prover one-round interactive proof system},
a verifier is a randomized polynomial-time process
having access to two provers (which we will call Alice and Bob).
All of the parties are given the same input string~$x$.
The verifier produces polynomial-length questions to Alice and Bob,
receives polynomial-length answers from them, and decides whether to
accept or reject.

A verifier~$V$ naturally defines a family of two-player one-round games
indexed by input strings.
A \emph{(classical) two-player one-round game}~$G=(S,T,Y,Z,\pi,R)$
is determined by finite sets~$S$, $T$, $Y$, and~$Z$,
a probability distribution~$\pi$ over~$S\times T$
and a function~$R\colon S\times T\times Y\times Z\to[0,1]$.
The value~$R(s,t,y,z)$ is written as~$R(y,z \mymid s,t)$ by convention.
This game is interpreted as a cooperative two-player game
of imperfect information played by two \emph{players} (Alice and Bob)
and run by a third party called the \emph{referee}, who enforces the
rules.
First the referee generates a pair of questions~$(s,t)\in S\times T$
according to the probability distribution~$\pi$
and sends~$s$ to Alice and~$t$ to Bob.
Then Alice responds to the referee with an element~$y\in Y$
and Bob responds with~$z\in Z$.
Finally the referee decides whether Alice and Bob win or lose, using
randomness in the most general situation:
Alice and Bob win with probability~$R(y,z\mymid s,t)$ and lose with
probability~$1-R(y,z\mymid s,t)$.
Note that if we fix a verifier and an input string~$x\in\{0,1\}^*$,
the verifier acts as a referee in some two-player one-round game~$G_{V,x}$.

A \emph{strategy} of players in a two-prover one-round
game~$G=(S,T,Y,Z,\pi,R)$ is a family of probability
distributions~$p_{s,t}$ over~$Y\times Z$ indexed
by~$(s,t)\in S\times T$, where the value~$p_{s,t}(y,z)$
represents the probability with which Alice replies with the
string~$y$ and Bob replies with the string~$z$ under the condition
that the verifier sends the question~$s$ to Alice and the question~$t$
to Bob.
It is customary to write~$p(y,z\mymid s,t)$ instead of~$p_{s,t}(y,z)$.
The strategy is said to be \emph{no-signaling}
if the following \emph{no-signaling conditions} are satisfied:
\begin{enumerate}
\item
  No-signaling from Alice to Bob:
  \[
    \sum_{y\in Y} p(y,z\mymid s,t) =
    \sum_{y\in Y} p(y,z\mymid s',t)
  \]
  for all $s,s'\in S$, $t\in T$, and $z\in Z$.
\item
  No-signaling from Bob to Alice:
  \[
    \sum_{z\in Z} p(y,z\mymid s,t) =
    \sum_{z\in Z} p(y,z\mymid s,t')
  \]
  for all $s\in S$, $t,t'\in T$, and $y\in Y$.
\end{enumerate}
For functions $a,b\colon\NN\to[0,1]$,
a two-prover one-round interactive proof system with a verifier~$V$ is
said to \emph{recognize} a promise problem~$A = (A_{\yes},A_{\no})$
with no-signaling provers with completeness probability at least~$a$
and soundness error at most~$b$ if the corresponding games satisfy
the following conditions:
\begin{itemize}
\item
  \emph{Completeness}.
  For every~$x\in A_{\yes}$, there exists a no-signaling strategy for
  the game~$G_{V,x}$ that makes the verifier accept with
  probability at least~$a(\abs{x})$.
\item
  \emph{Soundness}.
  For every~$x\in A_{\no}$, every no-signaling strategy for the
  game~$G_{V,x}$ makes the verifier accept with probability at
  most~$b(\abs{x})$.
\end{itemize}
The class of promise problems~$A$ having such a two-prover one-round
interactive proof system is denoted by~$\MIPns_{a,b}(2,1)$.
It is known that~$\MIPns_{a,b}(2,1)=\PSPACE$
for all polynomial-time computable functions~$a,b:\NN\to(0,1]$ for
which $a(n) - b(n) \geq 1/p(n)$ for some $p\in\poly$
\cite{ItoKobMat09CCC,Ito10ICALP}.

\subsection{\boldmath The proof system for $\EXP$ and its analysis}
  \label{section:two-prover}

This section describes a (classical) two-prover one-round interactive
proof system for $\EXP$ with perfect completeness (for uncorrelated
honest provers) and soundness error double-exponentially close to 1
against arbitrary no-signaling provers.
The proof system has the additional property that the verifier's
questions to the two provers are uniformly generated random strings,
which will be important in the next section.

For a Boolean circuit $C$ with $N$ gates $g_0,g_1,\dots,g_{N-1}$,
where gate~$g_j$ is an input to gate~$g_i$ only if~$j<i$,
a pair $(N,D)$ is called a \emph{succinct representation} of $C$
if $D$ is a Boolean circuit that, given an integer $0\le i\le N-1$,
returns the kind of gate~$g_i$ (ZERO, ONE, AND, OR, or NOT)
and the indices of gates from which the inputs to $g_i$ come (if any).
Note that a succinct representation of length~$n$
represents a Boolean circuit with at most~$2^n$ gates.
The \problemname{Succinct Circuit Value} problem
is the following decision problem.
\begin{center}
  \begin{minipage}{6in}
    \textbf{\problemname{Succinct Circuit Value}}\\[2mm]
    \begin{tabular}{@{}lp{5in}@{}}
    Instance: &
    A succinct representation of a Boolean circuit $C$ with $N$ gates
    whose fan-in is at most two and an integer $0\le k\le N-1$.\\[2mm]
    Question: &
    Does gate~$g_k$ have value~$1$?
    \end{tabular}
  \end{minipage}
\end{center}
The \problemname{Succinct Circuit Value} problem
is $\EXP$-complete (see, e.g., Theorem~3.31 of \cite{DuKo00}).
We will give a two-prover one-round interactive proof system
for \problemname{Succinct Circuit Value}
with the completeness and soundness conditions stated above.

\begin{theorem} \label{theorem:mipnsu}
  The \problemname{Succinct Circuit Value} problem has a two-prover
  one-round interactive proof system with no-signaling provers
  with perfect completeness and soundness error~$1-2^{-2^{p(n)}}$
  for some~$p\in\poly$, i.e.,
  \[
    \problemname{Succinct Circuit Value}
    \in\MIPns_{1,1-2^{-2^{\poly}}}(2,1).
  \]
  Moreover, for some constant~$\alpha>0$ and infinitely many input
  strings~$x$, the soundness error of this proof system is at
  least~$1-2^{-2^{\abs{x}^\alpha}}$.
\end{theorem}

\paragraph{Idea.}
The idea for the protocol is simple.
The honest provers hold the correct values of all gates in a circuit.
These values have to satisfy exponentially many local constraints,
and the verifier checks one of these local constraints chosen randomly.
It turns out that the local constraints, together with the no-signaling conditions,
are sufficient to restrict the value of each gate claimed by the provers
to the correct value inductively, beginning from the constant gates
and propagating from the inputs and the output of each gate,
concluding the soundness.

\paragraph{Protocol.}
Without loss of generality we assume that $N$ is a power of two
by adding unused gates as necessary.
The verifier chooses two integers $0\le s,t\le N-1$ uniformly and
independently.
He sends $s$ to Alice and $t$ to Bob.
Alice answers all the values of the input gates of $g_s$
in the same order as $D$ returns (if any).
Bob answers the value of $g_t$.
The verifier checks the following conditions.
\begin{enumerate}[(a)]
\item \label{enum:test-kind}
  If $s=t$, then Bob's answer must be equal to
  the value computed from Alice's answers (if any) and the kind of
  gate~$g_s$.
\item \label{enum:test-agree}
  If $g_t$ is an input to gate~$g_s$,
  then the value of $g_t$ claimed by Alice
  must agree with the value claimed by Bob.
\item \label{enum:test-output}
  If $t=k$, then Bob's answer must be $1$.
\end{enumerate}
The verifier accepts if and only if all the conditions~(a)--(c) are
satisfied.

\paragraph{Completeness.}
Completeness is easy: if the value of gate~$g_k$ is $1$, then provers
who simply answer the requested values of gates are accepted with
probability~$1$.

\paragraph{Soundness.}
Now we shall prove that this two-prover interactive proof system
has soundness error at most~$1-2^{-O(N)}=1-2^{-O(2^n)}$
against no-signaling dishonest provers.
Again we can assume that $N$ is a power of two without loss of generality.

Let $(N,D,k)$ be an instance of \problemname{Succinct Circuit Value},
and let~$v_i\in\{0,1\}$ be the value of gate~$g_i$
for~$0\le i\le N-1$.
Fix any no-signaling strategy in the two-prover interactive proof
system, and
let~$\varepsilon$ be the probability that this strategy is rejected.
We assume $\varepsilon<1/(N^2\cdot3^N)$
and prove that gate~$g_k$ has value~$1$.

Let~$\varepsilon(s,t)$ be the probability that this strategy
is rejected, conditioned on pair~$(s,t)$ of questions.
Then
\[
  \varepsilon=\frac{1}{N^2}\sum_{s,t}\varepsilon(s,t),
\]
which implies for any questions $s,t$,
it holds that
\[
  \varepsilon(s,t)\le\sum_{s',t'}\varepsilon(s',t')
  =N^2\varepsilon<\frac{1}{3^N}.
\]
Let~$\delta(i)$ be the probability
that Bob answers~$1-v_i$ when asked~$i$.

We prove that
\begin{equation}
  \delta(i)<\frac{3^i}{3^N}
  \label{eq:mipns-soundness}
\end{equation}
by induction on $i$.

First we consider the case where~$g_i$ is a constant gate.
This includes the case of~$i=0$.
As Bob gives a wrong answer with probability~$\delta(i)$
when Bob's question is~$i$, regardless of Alice's question,
$\delta(i)\le\varepsilon(i,i)$ by considering the probability
that the strategy fails in the test~(\ref{enum:test-kind}),
which implies
\[
  \delta(i)\le\varepsilon(i,i)<\frac{1}{3^N}\le\frac{3^i}{3^N}.
\]

Suppose $i\ge1$ and~$g_i$ is not a constant gate.
Assume $g_i$ is an AND or OR gate,
and let $j_1$ and $j_2$ be the indices of the inputs to $g_i$.
First consider Alice's answer in the case where her question is $i$.
If the value of $g_{j_1}$ claimed by Alice when her question is $i$ is wrong,
then when Bob's question is $j_1$,
either Bob's answer is wrong or Alice's and Bob's answers disagree.
If their answers disagree, then the verifier rejects by the test~(\ref{enum:test-agree}),
and therefore this happens with probability at most $\varepsilon(j_1,j_1)<1/3^N$.
As Bob's answer is wrong with probability~$\delta(j_1)$ and their
answers disagree with probability less than $1/3^N$, the value
of~$g_{j_1}$ claimed by Alice when her question is~$i$ is wrong with
probability at most
\[
  \delta(j_1)+\frac{1}{3^N}<\frac{3^{j_1}+1}{3^N}.
\]
In the same way,
the value of $g_{j_2}$ claimed by Alice when her question is $i$ is wrong
with probability at most
\[
  \delta(j_2)+\frac{1}{3^N}<\frac{3^{j_2}+1}{3^N}.
\]
If Bob's answer for $i$ is wrong,
then if both questions are $i$,
at least one of the following happens:
\begin{itemize}
\item
  The value of $g_{j_1}$ claimed by Alice is wrong.
  This happens with probability less than
  $(3^{j_1}+1)/3^N$.
\item
  The value of $g_{j_2}$ claimed by Alice is wrong.
  This happens with probability less than
  $(3^{j_2}+1)/3^N$.
\item
  The values of $g_{j_1}$ and $g_{j_2}$ claimed by Alice are correct,
  but the value of $g_i$ claimed by Bob is wrong.
  As this is detected by the test~(\ref{enum:test-kind}) of the
  verifier, it happens with probability at most $\varepsilon(i,i)<1/3^N$.
\end{itemize}
Therefore,
\[
  \delta(i)<\frac{3^{j_1}+1}{3^N}+\frac{3^{j_2}+1}{3^N}+\frac{1}{3^N}
  <\frac{3^i}{3^N}.
\]
The case where $g_i$ is a NOT gate is proved in a similar way.
This finishes the inductive case
and establishes the inequality~(\ref{eq:mipns-soundness})
for all~$i$.

The inequality~(\ref{eq:mipns-soundness})
implies that Bob's answer to question~$k$ is equal to~$v_k$
with probability greater than $1-3^k/3^N\ge2/3$.
On the other hand, by the test~(\ref{enum:test-output}),
Bob's answer to question~$k$ is equal to~$1$
with probability at least $1-\varepsilon(k,k)>1-1/3^N\ge2/3$.
These two conditions imply $v_k=1$.

\begin{remark}
  For a function~$a\colon\NN\to(0,1]$,
  let~$\MIPns_{a,<a}(2,1)$ denote the class of promise problems
  having a two-prover one-round interactive proof system
  with no-signaling provers
  with acceptance probability at least~$a$
  and soundness error strictly less than~$a$.
  Because the maximum acceptance probability for no-signaling provers
  can be computed exactly by solving an exponential-size linear
  program~\cite{Preda}, we have~$\MIPns_{a,<a}(2,1)\subseteq\EXP$
  for any polynomial-time computable function~$a\colon\NN\to(0,1]$
  by using any polynomial-time algorithm for linear programming~%
  \cite{Khachiyan79Doklady,Karmarkar84Comb}.
  Combined with Theorem~\ref{theorem:mipnsu},
  we have~$\MIPns_{a,<a}(2,1)=\EXP$ for any such~$a$.
\end{remark}

\paragraph{Tightness of soundness analysis.}
We shall prove the ``moreover'' part of Theorem~\ref{theorem:mipnsu}:
the double-exponential gap is tight for this protocol.
This will be used in the next section to prove
that the soundness error of the quantum interactive proof system
for~$\EXP$ that we construct is at least $1-2^{-2^{\poly}}$ on
infinitely many input strings.

This can be proved by studying the instance of the
\problemname{Succinct Circuit Value} problem used by Trevisan and
Xhafa~\cite{TreXha98PPL}.%
\footnote{%
  Note that we cannot avoid a large soundness error simply by
  restricting the problem to succinct Boolean formula values:
  with this restriction in place, the problem is in $\PSPACE$
  \cite{Lynch77JACM}.}
Let~$h$ be a positive integer.
Consider a circuit~$C$ with $N=2h+2$ gates
$g_0,g_1,\dots,g_{2h+1}$,
where $g_0$ and $g_1$ are ZERO gates
and, for $1\le i\le h$, $g_{2i}$ and $g_{2i+1}$ are two identical OR gates
whose inputs come from $g_{2(i-1)}$ and $g_{2(i-1)+1}$.
Clearly this circuit $C$ has a succinct representation
of length polylogarithmic in~$h$.
Let $k=2h-1$.

Alice and Bob decide their answers as follows.
First we describe each prover's marginal probability distribution.
When Bob is asked either $2i$ or $2i+1$ where $0\le i\le h$,
he answers $1$ with probability $1/2^{h-i}$
and $0$ with probability $1-1/2^{h-i}$.
When Alice is asked $2i$ or $2i+1$ where $1\le i\le h$,
she answers $(1,0)$ and $(0,1)$ each with probability $1/2^{h-i+1}$,
and $(0,0)$ with probability $1-2^{h-i}$.
The joint distribution of their answers is defined as follows.
In what follows, $(y_1,y_2;z)$ denotes that
Alice's answer is $(y_1,y_2)$ and Bob's answer is $z$.
\begin{itemize}
\item
  $s=t$, $\floor{s/2}=i\ge1$:
  Alice and Bob answer
  $(1,0;1)$ and $(0,1;1)$ each with probability $1/2^{h-i+1}$,
  and $(0,0;0)$ with probability $1-1/2^{h-i}$.
\item
  $\floor{s/2}=i\ge1$, $t=2(i-1)$:
  Alice and Bob answer
  $(1,0;1)$ and $(0,1;0)$ each with probability $1/2^{h-i+1}$,
  and $(0,0;0)$ with probability $1-1/2^{h-i}$.
\item
  $\floor{s/2}=i\ge1$, $t=2(i-1)+1$:
  Alice and Bob answer
  $(1,0;0)$ and $(0,1;1)$ each with probability $1/2^{h-i+1}$,
  and $(0,0;0)$ with probability $1-1/2^{h-i}$.
\item
  Otherwise:
  Alice and Bob give their answers in any way
  as long as the marginal distributions agree with the description above
  (e.g.\ they answer independently).
\end{itemize}

It is easy to check that this strategy is no-signaling.

With this strategy, the verifier accepts
unless $t\in\{0,1\}$ and Bob answers $1$
(which fails in test~(\ref{enum:test-kind})).
Therefore, the verifier accepts with probability at least
$1-1/((h+1)\cdot2^h)\ge1-2^{-h}=1-2^{-2^{n^\alpha}}$
for some constant~$\alpha>0$.

\section{Simulating no-signaling provers with quantum interactive proofs}
\label{sec:qipu}

In this section we present the second part of the proof of
the containment~(\ref{eq:qipu}), which is a simulation of the two-prover
one-round interactive proof system described in the previous section
by a quantum interactive proof system with perfect completeness and
unbounded soundness error.
The result in this section can be stated as the following lemma.

\begin{lemma} \label{lemma:mipns-qip}
  Let~$\varepsilon\colon\NN\to(0,1)$.
  Suppose that a promise problem~$A=(A_{\yes},A_{\no})$ has a
  two-prover one-round interactive proof system with no-signaling
  provers with perfect completeness and soundness error at
  most~$1-\varepsilon$.
  Assume moreover that, for each input~$x\in A_{\yes}\cup A_{\no}$,
  the verifier's questions are chosen uniformly at random from the set
  $\{0,1\}^{k(\abs{x})}\times\{0,1\}^{k(\abs{x})}$, for some
  function~$k\in\poly$.
  \begin{enumerate}[(i)]
  \item \label{enum:mipns-qip-a}
    It holds that $A\in\QIP(4,1,1-\varepsilon^2/144)$,
    that is, the problem $A$ has a four-message quantum interactive
    proof system with perfect completeness and soundness error at
    most~$1-\varepsilon^2/144$.%
    \footnote{%
      It is possible to replace the coeffieicnt~$1/144$ with a larger constant
      at the expense of introducing slight complications in several parts in the proof,
      but we will choose to use simpler arguments
      rather than trying to maximize the coefficient.}
  \item \label{enum:mipns-qip-b}
    If the original system has soundness error~$1-\varepsilon'$ on
    input~$x\in A_{\no}$, then the derived quantum interactive proof
    system has soundness error at least~$1-\varepsilon'/4$ on
    input~$x$.
  \end{enumerate}
\end{lemma}

Note that the containment~(\ref{eq:qipu}) follows by applying
Lemma~\ref{lemma:mipns-qip} to the two-prover one-round interactive
proof system for the \problemname{Succinct Circuit Value} problem
with no-signaling provers with perfect completeness and soundness
error at most~$1-2^{-2^{\poly}}$ constructed in the previous section.

\paragraph{Construction of the protocol.}
Given an input string~$x\in A_{\yes}\cup A_{\no}$, the verifier in the
quantum interactive proof system that we construct acts as follows.
First, the verifier prepares six quantum registers
$\sfS$, $\sfT$, $\sfS'$, $\sfT'$, $\sfY$, and $\sfZ$
in the state $\ket\Phi_{\sfS\sfS'}\ket\Phi_{\sfT\sfT'}\ket0_\sfY\ket0_\sfZ$,
where $\ket\Phi$ is the following maximally entangled state:
\[
  \ket\Phi=\left(\frac{\ket{00}+\ket{11}}{\sqrt2}\right)^{\otimes k},
\]
where~$k=k(\abs{x})$.
The four registers $\sfS$, $\sfT$, $\sfS'$, and $\sfT'$
are $k$ qubits long, and $\sfY$ and $\sfZ$ must be long enough to hold
Alice and Bob's answers in the two-prover one-round protocol.
Next, in the first round,
the verifier sends $\sfS$, $\sfT$, $\sfY$, and $\sfZ$
to the prover and the prover sends back the same registers.
Then, the verifier performs one of the following three tests
each with probability $1/4$, and accepts unconditionally with
probability $1/4$.
\begin{itemize}
\item
  \emph{Simulation test}:
  The verifier measures $\sfS'$, $\sfT'$, $\sfY$, and $\sfZ$ in the
  computational basis to obtain $s$, $t$, $y$, and $z$, respectively.
  If the result is accepted by the base two-prover protocol, then the
  verifier accepts; otherwise he rejects.
\item
  \emph{Undo-Alice test}:
  The verifier tells the prover that the undo-Alice test is to be
  performed.
  He then sends registers~$\sfS$ and $\sfY$ back to the prover, and
  receives $\sfS$.
  The verifier then destructively tests whether registers $\sfS$ and
  $\sfS'$ are in state $\ket\Phi$ or not.
  If they are, then he accepts; otherwise he rejects.
\item
  \emph{Undo-Bob test}:
  The verifier tells the prover that the undo-Bob test is to be
  performed.
  He then sends registers~$\sfT$ and $\sfZ$ back to the prover,
  and receives $\sfT$.
  The verifier then destructively tests whether registers $\sfT$ and
  $\sfT'$ are in state $\ket\Phi$ or not.
  If they are, then he accepts; otherwise he rejects.
\end{itemize}
Note that this verifier can be implemented exactly with the standard
Toffoli, Hadamard, $\pi/2$-phase-shift gate set.

\paragraph{Proof of completeness and part~(\ref{enum:mipns-qip-b}) of
  the lemma.}
Let~$x\in A_{\yes}\cup A_{\no}$.
We prove that if there exists a no-signaling strategy in the base
two-prover interactive proof system that makes the verifier accept
with probability~$1-\varepsilon'$, then the quantum interactive proof
system admits a strategy that makes the verifier accept with
probability~$1-\varepsilon'/4$.

Let~$p$ be the no-signaling strategy in the base two-prover
interactive proof system whose acceptance probability
is~$1-\varepsilon'$.
Let
\[
  p^\rmA(y\mymid s)=\sum_{z\in Z}p(y,z\mymid s,t), \qquad
  p^\rmB(z\mymid t)=\sum_{y\in Y}p(y,z\mymid s,t)
\]
be the marginal strategies,
which are well-defined because of the no-signaling conditions.
The prover in the constructed quantum interactive proof system
performs the following.
Registers~$\tilde{\sfS}$, $\tilde{\sfT}$, $\tilde{\sfY}$,
and~$\tilde{\sfZ}$ are the prover's private registers initialized
to~$\ket0$.
\begin{itemize}
\item
  In the first round, he performs the following operation on
  registers~$\tilde{\sfS}$, $\tilde{\sfT}$, $\sfY$, $\tilde{\sfY}$,
  $\sfZ$, and $\tilde{\sfZ}$ controlled on registers~$\sfS$ and~$\sfT$
  being in the state~$\ket{s}_{\sfS}\ket{t}_{\sfT}$:
  \[
    \ket0_{\tilde{\sfS}\tilde{\sfT}\sfY\tilde{\sfY}\sfZ\tilde{\sfZ}}
    \mapsto
    \ket{s}_{\tilde{\sfS}}\ket{t}_{\tilde{\sfT}}\sum_{y,z}\sqrt{p(y,z\mymid
      s,t)}\,\ket{yy}_{\sfY\tilde{\sfY}}\ket{zz}_{\sfZ\tilde{\sfZ}}.
  \]
  This controlled operation changes the global state as follows:
  \begin{align*}
    &
    \frac{1}{2^k}\sum_{s,t}\ket{ss0}_{\sfS\sfS'\tilde{\sfS}}\ket{tt0}_{\sfT\sfT'\tilde{\sfT}}\ket{00}_{\sfY\tilde{\sfY}}\ket{00}_{\sfZ\tilde{\sfZ}} \\
    &\mapsto
    \frac{1}{2^k}\sum_{s,t}\ket{sss}_{\sfS\sfS'\tilde{\sfS}}\ket{ttt}_{\sfT\sfT'\tilde{\sfT}}\sum_{y,z}\sqrt{p(y,z\mymid s,t)}\,\ket{yy}_{\sfY\tilde{\sfY}}\ket{zz}_{\sfZ\tilde{\sfZ}}.
  \end{align*}
\item
  In the undo-Alice test,
  he performs the following operation
  on registers~$\tilde{\sfS}$, $\sfY$, and~$\tilde{\sfY}$
  controlled on registers~$\sfS$, $\tilde{\sfT}$, and~$\tilde{\sfZ}$
  being in the state~$\ket{s}_{\sfS}\ket{t}_{\tilde{\sfT}}\ket{z}_{\tilde{\sfZ}}$:
  \[
    \ket{s}_{\tilde{\sfS}}\sum_y\sqrt{\frac{p(y,z\mymid s,t)}{p^\rmB(z\mymid t)}}\,\ket{yy}_{\sfY\tilde{\sfY}}
    \mapsto
    \ket{0}_{\tilde{\sfS}}\ket{00}_{\sfY\tilde{\sfY}},
  \]
  or does nothing if~$p^\rmB(z\mymid t)=0$.
  This controlled operation changes the global state to
  \[
    \frac{1}{2^k}\sum_{s,t}\ket{ss0}_{\sfS\sfS'\tilde{\sfS}}\ket{ttt}_{\sfT\sfT'\tilde{\sfT}}\ket{00}_{\sfY\tilde{\sfY}}\sum_z\sqrt{p^\rmB(z\mymid t)}\,\ket{zz}_{\sfZ\tilde{\sfZ}},
  \]
  which can be rewritten as
  \[
    \ket\Phi_{\sfS\sfS'}\ket0_{\tilde{\sfS}}\ket{00}_{\sfY\tilde{\sfY}}\otimes\frac{1}{\sqrt{2^k}}\sum_t\ket{ttt}_{\sfT\sfT'\tilde{\sfT}}\sum_z\sqrt{p^\rmB(z\mymid t)}\,\ket{zz}_{\sfZ\tilde{\sfZ}}
  \]
  by rearranging the registers.
\item
  In the undo-Bob test,
  he performs the following operation
  on registers~$\tilde{\sfT}$, $\sfZ$, and~$\tilde{\sfZ}$
  controlled on registers~$\tilde{\sfS}$, $\sfT$, and~$\tilde{\sfY}$
  being in the state~$\ket{s}_{\tilde{\sfS}}\ket{t}_{\sfT}\ket{y}_{\tilde{\sfY}}$:
  \[
    \ket{t}_{\tilde{\sfT}}\sum_z\sqrt{\frac{p(y,z\mymid s,t)}{p^\rmA(y\mymid s)}}\,\ket{zz}_{\sfZ\tilde{\sfZ}}
    \mapsto
    \ket{0}_{\tilde{\sfT}}\ket{00}_{\sfZ\tilde{\sfZ}},
  \]
  or does nothing if~$p^\rmA(y\mymid s)=0$.
  This controlled operation changes the global state to
  \[
    \frac{1}{2^k}\sum_{s,t}\ket{sss}_{\sfS\sfS'\tilde{\sfS}}\ket{tt0}_{\sfT\sfT'\tilde{\sfT}}\ket{00}_{\sfZ\tilde{\sfZ}}\sum_y\sqrt{p^\rmA(y\mymid s)}\,\ket{yy}_{\sfY\tilde{\sfY}},
  \]
  which can be rewritten as
  \[
    \ket\Phi_{\sfT\sfT'}\ket0_{\tilde{\sfT}}\ket{00}_{\sfZ\tilde{\sfZ}}\otimes\frac{1}{\sqrt{2^k}}\sum_s\ket{sss}_{\sfS\sfS'\tilde{\sfS}}\sum_y\sqrt{p^\rmA(y\mymid s)}\,\ket{yy}_{\sfY\tilde{\sfY}}.
  \]
\end{itemize}
This strategy passes the undo-Alice and undo-Bob tests with certainty,
and passes the simulation test with probability~$1-\varepsilon'$,
resulting in the overall acceptance probability~$1-\varepsilon'/4$.

In particular, this implies that this quantum interactive proof system
has perfect completeness
and the statement in part~(ii) of Lemma~\ref{lemma:mipns-qip}.

\paragraph{Proof of soundness.}
We prove the contrapositive: if there is a strategy in the
single-prover protocol that is accepted with high probability, then
the input must be a yes-instance.
Fix an instance~$x\in A_{\yes}\cup A_{\no}$ and a strategy in the
single-prover protocol that is accepted with probability~$1-\varepsilon'$,
where~$\varepsilon'<\varepsilon(\abs{x})^2/144$.
We prove that there is a no-signaling strategy for Alice and Bob
in the base two-prover protocol that is accepted with probability
at least~$1-12\sqrt{\varepsilon'}>1-\varepsilon(\abs{x})$,
implying that~$x\in A_{\yes}$.
As the verifier accepts with probability~$1-\varepsilon'$,
the simulation test, the undo-Alice test, and the undo-Bob test
each succeed with probability at least~$1-4\varepsilon'$.

Let register $\sfP$ denote the prover's private space.
Without loss of generality,
we assume that $\sfP$ is first initialized to $\ket0$
and that the prover performs
a unitary operation $U=U_{\sfS\sfT\sfY\sfZ\sfP}$
in the first round,
a unitary operation $V=V_{\sfS\sfY\sfP}$
in the second round in the undo-Alice test,
and a unitary operation $W=W_{\sfT\sfZ\sfP}$
in the second round in the undo-Bob test.
Let~$\ket\Psi$ be the state in registers~%
$\sfS$, $\sfT$, $\sfS'$, $\sfT'$, $\sfY$, $\sfZ$, and $\sfP$
after the first round:
\[
  \ket\Psi=(I_{\sfS'\sfT'}\otimes U_{\sfS\sfT\sfY\sfZ\sfP})
    \ket\Phi_{\sfS\sfS'}\ket\Phi_{\sfT\sfT'}\ket0_\sfY\ket0_\sfZ\ket0_\sfP.
\]

Let $\tilde{p}(s,t,y,z)$ be the probability with which the results of
the measurement in the simulation test are $s$, $t$, $y$, and $z$:
\[
  \tilde{p}(s,t,y,z)
  =\bra{s}_{\sfS'}\bra{t}_{\sfT'}\bra{y}_\sfY\bra{z}_\sfZ
   (\Tr_{\sfS\sfT\sfP}\ket\Psi\bra\Psi)
   \ket{s}_{\sfS'}\ket{t}_{\sfT'}\ket{y}_\sfY\ket{z}_\sfZ.
\]
Note that because the verifier never sends $\sfS'$ or $\sfT'$ to the prover,
the reduced state $\Tr_{\sfS\sfT\sfY\sfZ\sfP}\ket\Psi\bra\Psi$ is not affected
by the operation by the prover in the first round.
Therefore, $\Tr_{\sfS\sfT\sfY\sfZ\sfP}\ket\Psi\bra\Psi$
is the completely mixed state $I/2^{2k}$ on $\sfS'$ and $\sfT'$.
This implies $\sum_{y,z}\tilde{p}(s,t,y,z)=1/2^{2k}$ for every $s$ and $t$.
Let
\[
  p(y,z\mymid s,t)=2^{2k}\tilde{p}(s,t,y,z).
\]
We shall show that strategy~$p$ is ``close'' to some no-signaling strategy.
For this purpose, we use the notion of~$\delta$-no-signaling strategies.

A strategy~$p$ in a two-player one-round game is said to be \emph{$\delta$-no-signaling}
with respect to probability distribution~$\pi$ over the questions
if there exist single-prover strategies~$p^\rmA(y\mymid s)$ and $p^\rmB(z\mymid t)$
such that
\begin{align}
  \sum_{s,t} \pi(s,t) \frac12 \sum_y
  \abs*{\sum_z p(y,z\mymid s,t)-p^\rmA(y\mymid s)} &\le \delta,
    \label{eq:almost-nosig-a} \\
  \sum_{s,t} \pi(s,t) \frac12 \sum_z
  \abs*{\sum_y p(y,z\mymid s,t)-p^\rmB(z\mymid t)} &\le \delta.
    \label{eq:almost-nosig-b}
\end{align}
We will now prove that $p$ is $4\sqrt{\varepsilon'}$-no-signaling
with respect to the uniform distribution over the questions.
Toward this goal, we define
\[
  p^\rmA(y\mymid s)
  =
  \frac{1}{2^k} \sum_{t,z} p(y,z\mymid s,t),
  \qquad
  p^\rmB(z\mymid t)
  =
  \frac{1}{2^k} \sum_{s,y} p(y,z\mymid s,t),
\]
and prove the inequalities~(\ref{eq:almost-nosig-a})
and~(\ref{eq:almost-nosig-b})
with~$\delta=4\sqrt{\varepsilon'}$.

Let~$\rho$ be the state of registers~$\sfS'$, $\sfT$, $\sfT'$, and~$\sfY$
after the verifier receives a message from the prover
in the undo-Bob test:
\[
  \rho
  =
  \Tr_{\sfS\sfZ\sfP}
  (I_{\sfS\sfS'\sfT'\sfY} \otimes W_{\sfT\sfZ\sfP})
  \ket\Psi
  \bra\Psi
  (I_{\sfS\sfS'\sfT'\sfY} \otimes W_{\sfT\sfZ\sfP}^*).
\]
The fact that this strategy passes the undo-Bob test with probability at
least~$1-4\varepsilon'$ can be written as
\[
  1-
  \bra\Phi_{\sfT\sfT'}
  (\Tr_{\sfS'\sfY}\rho)
  \ket\Phi_{\sfT\sfT'}
  \le
  4\varepsilon'.
\]

We use the following easy lemma,
which will be proved at the end of this section.
In what follows, $\norm{X}_1$ denotes the trace norm of a matrix~$X$:
$\norm{X}_1=\Tr\sqrt{X^*X}$.

\begin{lemma} \label{lemma:pure-overlap}
  Let~$\calX$ and~$\calY$ be finite-dimensional Hilbert spaces.
  Then, for a pure state~$\ket\varphi\in\calX$
  and a density matrix~$\rho$ on~$\calX\otimes\calY$,
  it holds that
  \[
    \norm[\big]{\rho - \ket\varphi\bra\varphi \otimes \Tr_{\calX}\rho}_1
    \le
    4\sqrt{1-\bra\varphi(\Tr_{\calY}\rho)\ket\varphi}.
  \]
\end{lemma}

By Lemma~\ref{lemma:pure-overlap}, we have that
\[
  \norm[\big]{
    \rho
    -
    \ket\Phi\bra\Phi_{\sfT\sfT'}
    \otimes
    \Tr_{\sfT\sfT'}\rho
  }_1
  \le
  4\sqrt{4\varepsilon'}
  =
  8\sqrt{\varepsilon'}.
\]
Take the partial trace over~$\sfT$
and note that~$\Tr_{\sfT}\rho=\Tr_{\sfS\sfT\sfZ\sfP}\ket\Psi\bra\Psi$
to obtain that
\[
  \norm*{
    \Tr_{\sfS\sfT\sfZ\sfP}\ket\Psi\bra\Psi
    -
    \frac{I_{\sfT'}}{2^k}\otimes \Tr_{\sfS\sfT\sfT'\sfZ\sfP}\ket\Psi\bra\Psi
  }_1
  \le
  8\sqrt{\varepsilon'}.
\]
Note that
\begin{align*}
  \sum_z p(y,z\mymid s,t)
  &=
  2^{2k}
  \bra{s}_{\sfS'}\bra{t}_{\sfT'}\bra{y}_\sfY
  (\Tr_{\sfS\sfT\sfZ\sfP}\ket\Psi\bra\Psi)
  \ket{s}_{\sfS'}\ket{t}_{\sfT'}\ket{y}_\sfY,
  \\
  p^\rmA(y\mymid s)
  &=
  2^k
  \bra{s}_{\sfS'}\bra{y}_{\sfY}
  (\Tr_{\sfS\sfT\sfT'\sfZ\sfP}\ket\Psi\bra\Psi)
  \ket{s}_{\sfS'}\ket{y}_{\sfY}.
\end{align*}
Then,
\begin{align*}
  &
  \frac{1}{2^{2k}}\sum_{s,t} \frac12 \sum_y \abs*{\sum_z p(y,z\mymid s,t)-p^\rmA(y\mymid s)}
  \\
  &=
  \frac12\sum_{s,t,y}\abs*{
    \bra{s}_{\sfS'}\bra{t}_{\sfT'}\bra{y}_{\sfY}
    \left(
      \Tr_{\sfS\sfT\sfZ\sfP}\ket\Psi\bra\Psi
      -
      \frac{I_{\sfT'}}{2^k}\otimes \Tr_{\sfS\sfT\sfT'\sfZ\sfP}\ket\Psi\bra\Psi
    \right)
    \ket{s}_{\sfS'}\ket{t}_{\sfT'}\ket{y}_{\sfY}
  }
  \\
  &\le
  \frac12\norm*{
    \Tr_{\sfS\sfT\sfZ\sfP}\ket\Psi\bra\Psi
    -
    \frac{I_{\sfT'}}{2^k}\otimes \Tr_{\sfS\sfT\sfT'\sfZ\sfP}\ket\Psi\bra\Psi
  }_1
  \\
  &\le
  4\sqrt{\varepsilon'},
\end{align*}
and therefore the inequality~(\ref{eq:almost-nosig-a}) is satisfied.
The proof of  the inequality~(\ref{eq:almost-nosig-b}) is analogous.
This establishes the claim that strategy~$p$
is~$4\sqrt{\varepsilon'}$-no-signaling.

Now we prove that a $\delta$-no-signaling strategy is close to some
no-signaling strategy.
We use a property of the no-signaling conditions shown by
Holenstein~\cite{Holenstein09TOC}.
By applying Lemma~9.4 in~\cite{Holenstein09TOC} twice,
we obtain the following.

\begin{lemma} \label{lemma:almost-nosig-nearly-nosig}
  Let~$p$ be a~$\delta$-no-signaling strategy
  with respect to a probability distribution~$\pi$.
  Then there exists a no-signaling strategy~$\hat{p}$ such that
  \[
    \sum_{s,t} \pi(s,t) \frac12 \sum_{y,z}
    \abs{p(y,z\mymid s,t)-\hat{p}(y,z\mymid s,t)} \le 2\delta.
  \]
\end{lemma}

\noindent
The proof of Lemma~\ref{lemma:almost-nosig-nearly-nosig}
is the same as that of Lemma~9.5 in~\cite{Holenstein09TOC}, and is omitted.

By Lemma~\ref{lemma:almost-nosig-nearly-nosig},
there exists a no-signaling strategy~$\hat{p}$
such that
\[
  \frac{1}{2^{2k}}\sum_{s,t} \frac12 \sum_{y,z}
  \abs{p(y,z\mymid s,t)-\hat{p}(y,z\mymid s,t)} \le 8\sqrt{\varepsilon'}.
\]
As the simulation test succeeds with probability at
least~$1-4\varepsilon'$, the no-signaling strategy~$\hat{p}$ makes the
verifier in the base two-prover protocol accept with probability at
least
\[
  1-4\varepsilon'-8\sqrt{\varepsilon'}
  \ge
  1-12\sqrt{\varepsilon'}
  >
  1-\varepsilon(\abs{x}).
\]
By the soundness of the base two-prover protocol, it must hold
that~$x\in A_{\yes}$.
Therefore, the quantum interactive proof has soundness error at
most~$1-\varepsilon^2/144$.

In the rest of the section,
we will prove Lemma~\ref{lemma:pure-overlap}.
We use the following variant
of Winter's gentle measurement lemma~\cite{Winter99IEEEIT},
proved by Ogawa and Nagaoka~\cite{OgaNag07IEEEIT}.

\begin{lemma}
    \label{lemma:gentle-measurement}
  Let~$\calH$ be a finite-dimensional Hilbert space.
  For a density matrix~$\rho$ on~$\calH$
  and a Hermitian matrix~$A$ on~$\calH$
  such that both~$A$ and~$I_{\calH}-A$ are positive semidefinite,
  it holds that
  \[
    \norm{\rho-\sqrt{A}\,\rho\sqrt{A}}_1
    \le
    2\sqrt{\Tr\rho(I_{\calH}-A)}.
  \]
\end{lemma}

\begin{proof}[Proof of Lemma~\ref{lemma:pure-overlap}]
  Let~$Y=(\bra\varphi \otimes I_{\calY})\rho(\ket\varphi \otimes I_{\calY})$,
  and let~$A=\ket\varphi\bra\varphi\otimes I_{\calY}$.
  Then, it holds that
  \begin{align*}
    \sqrt{A}\,\rho\sqrt{A}
    &=
    A\rho A
    =
    \ket\varphi\bra\varphi \otimes Y,
    \\
    \Tr\rho A
    &=
    \bra\varphi(\Tr_{\calY}\rho)\ket\varphi.
  \end{align*}
  By Lemma~\ref{lemma:gentle-measurement}, it holds that
  \[
    \norm[\big]{\rho - \ket\varphi\bra\varphi \otimes Y}_1
    \le
    2\sqrt{1-\bra\varphi(\Tr_{\calY}\rho)\ket\varphi}.
  \]
  which implies that
  \[
    \norm{\Tr_{\calX}\rho - Y}_1
    \le
    2\sqrt{1-\bra\varphi(\Tr_{\calY}\rho)\ket\varphi}.
  \]
  Then we have that
  \begin{align*}
    &
    \norm[\big]{
      \rho
      -
      \ket\varphi\bra\varphi \otimes \Tr_{\calX}\rho
    }_1
    \\
    &\le
    \norm[\big]{
      \rho
      -
      \ket\varphi\bra\varphi \otimes Y
    }_1
    +
    \norm[\big]{
      \ket\varphi\bra\varphi \otimes \Tr_{\calX}\rho
      -
      \ket\varphi\bra\varphi \otimes Y
    }_1
    \\
    &=
    \norm[\big]{
      \rho
      -
      \ket\varphi\bra\varphi \otimes Y
    }_1
    +
    \norm{
      \Tr_{\calX}\rho
      -
      Y
    }_1
    \\
    &\le
    4\sqrt{1-\bra\varphi(\Tr_{\calY}\rho)\ket\varphi}.
    \qedhere
  \end{align*}
\end{proof}

\section{Additional results}
\label{sec:additional-results}

In this section we mention some additional results about quantum
interactive proof systems with unbounded error.

\subsection{\boldmath One-round quantum interactive proofs for $\PSPACE$ with a weak error bound}

\begin{theorem} \label{theorem:qipu2}
  It holds that $\PSPACE\subseteq\QIP(2,1,1-2^{-\poly})$.
\end{theorem}

\begin{proof}
  The \problemname{Succinct Bipartiteness} problem is the problem of
  deciding if an exponential-size graph, given in its succinct
  representation, is bipartite.
  It is known to be $\PSPACE$-complete~\cite{LozBal89WG}.
  It is straightforward to construct a two-prover one-round XOR
  interactive proof system with perfect completeness and an
  exponentially small gap for this problem.
  (We refer the reader to \cite{CleHoyTonWat04CCC,Wehner06STACS}
   for the definition of XOR interactive proof systems.)
  This proves the
  containment
  \[
    \PSPACE\subseteq\xorMIP_{1,1-2^{-\poly(n)}}[2].
  \]
  Theorem~5.10 of Cleve, H{\o}yer, Toner, and
  Watrous~{\cite{CleHoyTonWat04CCC}} implies that
  \[
    \xorMIP_{1,1-2^{-\poly(n)}}[2]=\xorMIPstar_{1,1-2^{-\poly(n)}}[2],
  \]
  and the construction of Wehner~\cite{Wehner06STACS} implies
  \[
    \xorMIPstar_{1,1-2^{-\poly(n)}}[2]\subseteq
    \QIP(2,1,1-2^{-\poly}).
  \]
  We obtain the theorem by chaining these inclusions.
\end{proof}

\subsection{Upper bounds}

One may also consider the power of quantum interactive proof systems
when acceptance is defined by a sharp threshold value.
That is, for any choice of functions $m\in\poly$ and $a\colon\NN
\rightarrow (0,1]$, we may consider the class $\QIP(m,a,<\!a)$,
defined as the class of promise problems $A = (A_{\yes},A_{\no})$
having a quantum interactive proof system with $m(\abs{x})$ messages
that accepts with probability at least $a(\abs{x})$ on inputs $x\in
A_{\yes}$, and with probability strictly smaller than $a(\abs{x})$
on all inputs $x\in A_{\no}$.
The notation $\QMA(1,<\!1)$ is shorthand for $\QIP(1,1,<\!1)$.
The following two theorems concerning these classes are proved.

In this section, the following mild assumptions are made on the gate set:
\begin{itemize}
\item
  The gate set consists of a finite number of gates.
\item
  The amplitudes of each gate in the gate set are algebraic numbers.
\end{itemize}
Without the second restriction,
even $\BQP$ would contain some undecidable languages;
see Theorem~5.1 of Adleman, Demarrais, and Huang~\cite{AdlDemHua97SICOMP}.

\subsubsection{Upper bound on $\QIP(\poly,a,<\!a)$}

\begin{theorem} \label{theorem:upper-qip}
  For any polynomial-time computable function $a:\NN\rightarrow(0,1]$,
  it holds that
  \[
    \QIP(\poly,a,<\!a)\subseteq\EXPSPACE.
  \]
\end{theorem}

As stated in the introduction,
Gutoski and Watrous~\cite{GutWat07STOC} give a semidefinite program
representing the optimal acceptance probability
of a given quantum interactive proof system.
When applied to the class~$\QIP(\poly,a,<\!a)$, with our relaxed
assumptions on the gate set, this transformation results in a
semidefinite program of exponential size with algebraic coefficients.
The remaining task is to decide whether this semidefinite program
has the optimal value at least~$a$ or less than~$a$.
This task can be formulated as an exponential-size instance
of the \problemname{Semidefinite Feasibility with Algebraic
Coefficients} problem.

\problemname{Semidefinite Feasibility with Algebraic Coefficients}
is a problem based on semidefinite programming.
Let $\QQ$, $\RR$, and $\bar\QQ\cap\RR$
be the fields of rational numbers, real numbers, and algebraic real numbers,
respectively.
Each element $\alpha$ of $\bar\QQ\cap\RR$ can be encoded
as a triple $(f(X),a,b)$ of
the minimum polynomial $f(X)$ of $\alpha$ over $\QQ$
and $a,b\in\QQ$ with $a<\alpha<b$
such that $\alpha$ is the only root of $f(X)$ between $a$ and $b$.
(See Section~10.2 of Basu, Pollack, and Roy~\cite{BasPolRoy03}.)
\begin{center}
  \begin{minipage}{6in}
    \textbf{\problemname{Semidefinite Feasibility with Algebraic
        Coefficients}}\\[2mm]
    \begin{tabular}{@{}lp{5in}@{}}
    Instance: &
    Integers $n,d>0$,
    $m$ algebraic real matrices $A_1,\dots,A_m$ of size $d\times d$,
    and $m$ algebraic real numbers $b_1,\dots,b_m$.\\[2mm]
    Question: &
    Does there exist a $d\times d$ real matrix $X\succeq0$
    such that $\Tr A_i X=b_i$ for all $i$?
    \end{tabular}
  \end{minipage}
\end{center}
The complexity of
\problemname{Semidefinite Feasibility with Algebraic Coefficients} is
not known.
(See Ramana~\cite{Ramana97MP} for related results.)
Although there exist polynomial-time algorithms for semidefinite
programming that compute an approximate solution to an arbitrary
precision, they cannot be applied in a straightforward way to the
\problemname{Semidefinite Feasibility with Algebraic Coefficients}
problem.
We point out that the problem is in $\PSPACE$ by using the following
result.

\begin{theorem}[Canny~\cite{Canny88STOC}]
  \label{theorem:etr}
  The problem \problemname{Existential Theory of the Reals} is
  in~$\PSPACE$.
  That is, given a quantifier-free Boolean formula~$F(x_1,\dots,x_k)$
  with atomic predicates of the forms $p(x_1,\dots,x_k)=0$ and
  $p(x_1,\dots,x_k)>0$, where~$p$ is a polynomial with integer
  coefficients given as a list of coefficients in binary notation,
  it is decidable in space polynomial in the length of the formula~$F$
  whether there exists~$(x_1,\dots,x_k)\in\RR^k$
  that satisfies~$F$.
\end{theorem}

\begin{corollary}
  The problem \problemname{Semidefinite Feasibility with Algebraic Coefficients}
  is in $\PSPACE$.
\end{corollary}

\begin{proof}
  Note that an algebraic number
  encoded as~$(f(X),a,b)$
  can be represented as a variable~$x$
  constrained as~$f(x)=0\wedge x-a>0\wedge b-x>0$.
  By using this, we can write down each linear constraint~$\Tr A_i X=b_i$
  in terms of the variables representing the~$d^2$ coordinates of~$X$.
  Moreover, the semidefinite constraint~$X\succeq0$ can be written
  as~$\exists M.X=M^{\trans}M$,
  and therefore can be written as polynomial constraints
  on the coordinates of~$X$.
\end{proof}

By combining the semidefinite programming formulation of~\cite{GutWat07STOC}
and the polynomial-space algorithm
for \problemname{Semidefinite Feasibility with Algebraic Coefficients},
we obtain Theorem~\ref{theorem:upper-qip}.

\subsubsection{Upper bound on $\QMA(1,<\!1)$}

\begin{theorem} \label{theorem:upper-qma}
  It holds that $\QMA(1,<\!1)\subseteq\PSPACE$.
\end{theorem}

\begin{proof}
  Let $L\in\QMA(1,<\!1)$.
  The same technique as the proof of $\QMA\subseteq\PP$
  by Marriott and Watrous~\cite{MarWat05CC}
  reduces $L$ to a problem of deciding whether or not an implicitly
  given exponential-sized matrix $A$ has an eigenvalue~$1$,
  or equivalently $I-A$ is singular.

  The entries of $A$ are in a field $F$ that depends on the language
  $L$ as follows.
  Let $\alpha_1,\dots,\alpha_u\in\CC$
  be the distinct numbers that appear as entries in the natural
  representations of the gates in the gate set used by the verifier
  in the system for the language $L$.
  Let $F=\QQ(\alpha_1,\dots,\alpha_u)$
  be the field generated by the adjunction of $\alpha_1,\dots,\alpha_u$
  to the field $\QQ$,
  i.e.\ the smallest field containing all the rational numbers
  and $\alpha_1,\dots,\alpha_u$.
  Because $\alpha_1,\dots,\alpha_u$ are algebraic, $F$ is a finite
  extension of the field $\QQ$.
  By the primitive element theorem (see e.g.\ Problem~7.5 of \cite{Lorenz06}),
  there exists an algebraic number $\alpha\in F$
  such that $F=\QQ(\alpha)$.
  Let $f(t)$ be the minimal polynomial of $\alpha$ over $\QQ$
  and $d$ be the degree of $f(t)$.
  The field $F$ is isomorphic to the quotient field $\QQ[t]/(f(t))$,
  by which we identify $F$ with the set of polynomials over $\QQ$
  of degree at most $d-1$.
  Using this representation,
  addition, subtraction, multiplication, division, and equality testing
  of the numbers in $F$ can be performed in $\NC$.

  Using this representation,
  each entry of $A$ can be computed in $\PSPACE$.
  Csanky's algorithm~\cite{Csanky76SICOMP} can then be used to
  determine whether $I-A$ is singular or not in $\PSPACE$.
\end{proof}

\section{Open problems}

We conclude with a short list of open problems related to quantum
interactive proof systems with an unbounded error.
\begin{itemize}
\item
  Is $\EXP\subseteq\QIP(2,1,<\!1)$?
\item
  We have~$\PSPACE\subseteq\QIP(\poly,1,1-2^{-\poly})\subseteq\EXP$.
  Where does~$\QIP(\poly,1,1-2^{-\poly})$ lie?
  One may try to prove~$\QIP(\poly,1,1-2^{-\poly})=\PSPACE$
  by improving the dependence of the parallel time of an approximation
  algorithm for semidefinite programming on the error parameter.
  Note, however, that this is open even for the special case of
  positive linear programming~\cite{TreXha98PPL}.
\item
  Is it possible to improve our upper bound of $\EXPSPACE$ on
  $\QIP(\poly,a,<\!a)$?
  In particular, is it possible to avoid resorting to the exact feasibility
  of a semidefinite program?
  Or does the succinct version of the semidefinite feasibility problem
  belong to~$\QIP(\poly,a,<\!a)$?
  How small can the gap in acceptance probability
  between the completeness case and the soundness case be
  in a quantum interactive proof system?
\item
  Does the containment~$\QMA\subseteq\PP$~\cite{MarWat05CC}
  extend to the unbounded-error case?
  Our upper bound of $\PSPACE$ may not hold if perfect completeness is
  not assumed.
\end{itemize}

\section*{Acknowledgments}

We thank the anonymous reviewers of an earlier version of this paper
for helpful comments.
Tsuyoshi Ito acknowledges support from NSERC, CIFAR, QuantumWorks,
MITACS, CFI, and ORF\@.
Hirotada Kobayashi is partially supported by the Grant-in-Aid for
Scientific Research~(B) No.~21300002 of the Japan Society for the
Promotion of Science.
John Watrous acknowledges support from NSERC, CIFAR, and MITACS\@.

\bibliography{unbounded-qip}

\end{document}